\documentclass[reqno,12pt]{amsart}

\IfFileExists{mymtpro2.sty}{%
  \usepackage[subscriptcorrection]{mymtpro2}
}{}

\usepackage{a4,amssymb}

\usepackage{appendix}
\newtheorem{theorem}{Theorem}[section]
\newtheorem{lemma}{Lemma}[section]
\newtheorem{corollary}{Corollary}[section]
\newtheorem{proposition}{Proposition}[section]
\theoremstyle{definition}
\newtheorem{remark}[theorem]{Remark}

\newcommand{\labelnummer}{\mbox{\normalfont (\roman{numcount})}}%

\makeatletter

  {\let\curlabelspeicher\@currentlabel%
    \begin{list}{\labelnummer}%
      {\usecounter{numcount}\leftmargin0pt%
        \topsep0.5ex\partopsep2ex\parsep0pt\itemsep0ex\@plus1\p@%
        \labelwidth2.5em\itemindent3.5em\labelsep1em%
      }%
    \let\saveitem\item%
    \def\item{\saveitem%
      \def\@currentlabel{{\upshape\curlabelspeicher}$\,$\labelnummer}}%
    \let\savelabel\label%
    \def\label##1{\savelabel{##1}%
      \@bsphack%
        \ifmmode\else%
          \protected@write\@auxout{}%
          {\string\newlabel{##1item}{{\labelnummer}{\thepage}}}%
        \fi%
      \@esphack%
    }%
  }{\end{list}}%

\renewcommand{\appendix}{\def\thesection{\textsc{Appendix}}}

 \let\leq\le
 \let\geq\ge

\DeclareMathOperator{\tr}{tr\kern1pt}

\makeatletter

\newif\ifper\pertrue
\def\per{.}

\def\bti{\@ifnextchar[\bbti\bbbti}
\def\bbti[#1]#2{#2, #1.}
\def\bbbti#1{#1.}

\def\z{\@ifnextchar[\zz\zzz}
\def\zz[#1]#2#3#4#5{\perfalse\emph{#2} \textbf{#3}, #4 (#5) [#1]}
\def\zzz#1#2#3#4{\emph{#1} \textbf{#2}, #3 (#4)\ifper\per\fi\pertrue}

\def\pub{\@ifstar\pubstar\pubnostar}
\def\pubnostar{\@ifnextchar[\@@pubnostar\@pubnostar}
\def\@@pubnostar[#1]#2#3#4{#2, #3, #4, #1\ifper\per\fi\pertrue}
\def\@pubnostar#1#2#3{#1, #2, #3\ifper\per\fi\pertrue}
\def\pubstar[#1]#2#3#4{\perfalse #2, #3, #4 [#1]\pertrue}

\makeatother

\newcommand{\beq}{\begin{equation}}
\newcommand{\eeq}{\end{equation}}
\newcommand{\ba}{\begin{array}}
\newcommand{\ea}{\end{array}}
\newcommand{\bea}{\begin{eqnarray}}
\newcommand{\eea}{\end{eqnarray}}

{

\newcommand{\R}{\mathbb{R}}

\newcommand{\Z}{\mathbb{Z}}

\newcommand{\N}{\mathbb{N}}

\DeclareMathOperator{\supp}{\mathrm{supp}}

\def\P{I\kern-.30em{P}}
\def\E{I\kern-.30em{E}}
\renewcommand{\E}{\mathbb{E}\mkern2mu}
\renewcommand{\P}{\mathbb{P}}

\begin{document}

\title[Spectral Averaging and the Density of States]{Some remarks on spectral averaging and the local density of states for random Schr\"odinger operators on $L^2 ( \R^d )$}

\author[J.\ M.\ Combes]{Jean Michel Combes}
\address{Centre de Physique Th\'eorique \\
CNRS Luminy Case 907 \\
Marseille Cedex 9 \\
13009 Marseille France}
\email{jmcombes@cpt.fr}

\author[P.\ D.\ Hislop]{Peter D.\ Hislop}
\address{Department of Mathematics,
    University of Kentucky,
    Lexington, Kentucky  40506-0027, USA}
\email{peter.hislop@uky.edu}

\thanks{PDH thanks the Centre de Physique Th\'eorique, CNRS, Marseille, France, for support during the time this work was begun.}

\begin{abstract}
We prove some local estimates on the trace of spectral projectors for random Schr\"odinger operators restricted to cubes $\Lambda \subset \R^d$. We also present a new proof of the spectral averaging result based on analytic perturbation theory. Together, these provide another proof of the Wegner estimate with an explicit form of the constant and an alternate proof of the Birman-Solomyak formula. We also use these results to prove the Lipschitz continuity of the local density of states function for a restricted family of random Schr\"odinger operators on cubes $\Lambda \subset \R^d$, for $d \geq 1$. The result holds for low energies without a localization assumption but is not strong enough to extend to the infinite-volume limit.
\end{abstract}

\maketitle \thispagestyle{empty}

\begin{center}
Dedicated to the memory of Erik Baslev
\end{center}

\tableofcontents

%%%%%%%%%%%%%%%%%%%%%%%%%%%%%%%%%%%%%%%%%%%%%%%%%%%%%%%%%%%%%%%%%%%%%%%%%%%%%%%%%%%%%%%%%%%%%%%%%%%%%%%%%%%

\section{Statement of the Problem and Result}\label{sec:introduction}
\setcounter{equation}{0}

This note is another presentation of spectral averaging with applications to the study of the local density of states ($\ell$DOS) for random Schr\"odiner operators on cubes $\Lambda \subset \R^d$, for $d \geq 1$. Spectral averaging is revisited using tools from analytic perturbation theory, an area in which Erik Balslev was an expert. We also prove an upper bound on the trace of spectral projectors using a Poincar\'e-type inequality for eigenfunctions. We present three applications: 1) we prove the Wegner estimate with an explicit form of the constant, 2) we prove that the local density of states function is Lipschitz continuous in the energy, independent of localization, and 3) we  give  a simple proof of the Birman-Solomyak Theorem. The spectral averaging result applies to self-adjoint operators of the form $H_\omega = H_0 + \omega u^2$ on a separable Hilbert space where $H_0$ has discrete spectrum. A version of the Birman-Solomyak formula for the spectral shift function is proved in this setting.
%Spectral averaging is revisited using tools from analytic perturbation theory, an area in which Erik was an expert.

%The second application is to the density of states for random Scr\"odinger operators.
 The random Schr\"odinger operators that we study in the applications have the form
\beq\label{eq:rso1}
H_\omega := H_0 +  V_\omega,
\eeq
on $L^2 (\R^d)$, where $H_0$ is a self-adjoint operator such as the Laplacian $H_0 = - \Delta$ or a magnetic Schr\"odinger operator,  and the potential $V_\omega$ is a random, ergodic process described as follows.

\vspace{.1in}
\noindent
\textbf{Hypothesis 1 [H1]. }Single-site potential: 
 Let $u_0(x) \in L^\infty_0 (\R^d; \R)$ be a compactly-supported function satisfying 
$$
0 \leq \kappa \chi_0 \leq u_0^2 \leq 1 , 
$$
 for some $\kappa > 0$, and where $\chi_0$ is the characteristic function on the unit cube $C_0 := [0,1]^d$. 

\vspace{.1in}
\noindent
\textbf{Hypothesis 2 [H2].} Random variables: Let $\omega := \{ \omega_k \}_{k \in \Z^d}$ denote a family of independent, identically distributed (iid) random variables with $\omega_0 \geq 0$ with common probability density $\rho$ having compact support. 

\vspace{.1in}

%\noindent
For $k \in \Z^d$, we denote by $u_k$ the translate of $u_0$ by $k$, that is, $u_k(x) := u_0(x-k)$. Similarly, $C_k$ denotes the translation of $C_0$ by  $k \in \Z^d$ and we write $\chi_k$ for the characteristice function on the unit cube $C_k$. 
%We let $\omega := ( \omega_k)_{k \in \Z^d}$ denote a family of independent, identically distributed (iid) random variables with $%\omega_0 \geq 0$ with common probability density $\rho$ having compact support. 
The random potential $V_\omega$ is defined to be
\beq\label{eq:pot1}
V_\omega (x) := \sum_{k \in \Z^d} ~ \omega_k u_k(x) .
\eeq
We work with a restricted version of the random potential in section \ref{sec:LipCont1}: 

\vspace{.1in}
\noindent
\textbf{Hypothesis 3 [H3]. } The single-site potential $u_0 =  \kappa \chi_{0}$, for some $\kappa > 0$, where $\chi_0$ is the characteristic function of the unit cube $C_0 := [0,1]^d$. The single-site probability measure is the uniform measure on the interval $[0.1]$

\vspace{.1in}
We need local operators $H_\omega^\Lambda$ obtained from $H_\omega$ by restricting to cubes $\Lambda_L := [-L, L]^d$,  for $L \in \N$, and imposing self-adjoint boundary conditions, such as Dirichlet, Neumann, or periodic boundary conditions. The unperturbed operator $H_0^\Lambda$ is associated with the nonnegative quadratic form:
\beq\label{eq:quadform1}
f \in Q (\Lambda) \rightarrow \int_{\Lambda} ~|\nabla f|^2,
\eeq
for $f$ in the appropriate form domain $Q(\Lambda)$ determined by the boundary conditions. From the construction of $V_\omega$, this potential is relatively $H_0$-bounded with relative $H_0$ bound less than one, so $H_\omega^\Lambda$ is self-adjoint on the same domain as $H_0^\Lambda$.  Furthermore, $H_\omega^\Lambda$ has a compact resolvent so the spectrum of $H_\omega^\Lambda$ is discrete. We write  $ P_\omega^\Lambda(I)$ for the spectral projector for $H_\omega^\Lambda$ and the interval $I \subset \R$.

The $\ell$DOS measure $\mu_\Lambda$ is defined as the  number of eigenvalues of $H_\omega^\Lambda$ in the interval $I  = [I_-, I_+] \subset \R$ per unit volume:
\beq\label{eq:dos-defn1}
\mu_\Lambda (I) := \frac{1}{|\Lambda|} \E \{ {\rm Tr} P_\omega^\Lambda(I) \}  .
\eeq
The density of states measure for the infinite-volume operator $H_\omega$ is obtained 
by taking $|\Lambda| \rightarrow \infty$. It exists almost surely, see, for example, \cite{kirsch}.
%where $P_\omega^\Lambda (I)$ is the spectral projector for $H_\omega^\Lambda$ and interval $I \subset \R$. 
%for any interval $I \subset \R$.
The Wegner estimate \cite{CHK2-2007} in this setting is the  bound
$$
\E \{ {\rm Tr} P_\omega^\Lambda(I) \} = |\Lambda| \mu_\Lambda (I) \leq C_W (I_+) | \Lambda | |I| .
$$
%where $ P_\omega^\Lambda(I)$ is the local spectral projector for $H_\omega^\Lambda$ and the interval $I$.
%  = [I_-, I_+] \subset \R$.
This bound shows that the measure $\mu_\Lambda$ is absolutely continuous with respect to Lebesgue measure. The locally bounded density of the $\ell$DOS measure is denoted $n_\Lambda (E)$.

\subsection{Contents}

In section \ref{sec:dos-poincare1}, we prove an upper bound on the trace of a spectral projector of a local Schr\"odinger operator. The upper bound is expressed in terms of the matrix elements of the spectral projector with respect to the eigenfunctions of thel Neumann Laplacian of the unit cube. The spectral averaging result is derived in section \ref{sec:spectralAve1} using analytic perturbation theory for one-parameter families of self-adjoint operators. An application is given relating the spectral shift function to the local DOS proving a form of the Birman-Solomyak formula. Finally, in section \ref{sec:LipCont1}, we prove the local Lipschitz continuity of the DOS for random Schr\"odinger operators restricted to finite domains.

%%%%%%%%%%%%%%%%%%%%%%%%%%%%%%%%%%%%%%%%%%%%%%%
\section{Trace estimates from the Poincar\'e inequality}\label{sec:dos-poincare1}
\setcounter{equation}{0}

Let $h_{0,k}$ denote the Neumann Laplacian on the unit cube $C_k \subset \R^d$ that is the translate of the unit cube $C_0 := [0,1]^d$ by $k \in \Z^d$. The $L^2$-eigenfunctions of the self-adjoint operator $h_{0,k}$ are
$\varphi_{j,k}$ with eigenvalues $E_{j,k}$, listed including multiplicity. The set of eigenvalues is $\{  ( \sum_{m=1}^d n_m^2 ) \pi^2 ~|~ (n_1, \ldots, n_d) \in \{0,1,2, \ldots\}^d \}$. The set of eigenfunctions $\{ \varphi_{j,k} \}$ forms an orthonormal basis of $L^2 ( C_k)$. The spectral representation of $h_{0,k}$ is
$$
h_{0,k} = \sum_{j=0}^\infty E_{j,k} \Pi_{\varphi_{j,k}} ,
$$
where $\Pi_{\varphi_{j,k}}$ is the projection onto the vector $\varphi_{j,k} \in L^2 (C_k)$. In general, we let $\Pi_\psi$ denote the projection onto $\psi$ in the appropriate Hilbert space. 

In the following, we denote by $\Lambda$ the cube $\Lambda := [-L,L]^d$, with $L \in \N$, and we denote by $\widetilde{\Lambda}$ the integer lattice points in $\Lambda$ so that $\widetilde{\Lambda} := \Lambda \cap \Z^d$.

\begin{theorem}\label{thm:poincare1} 
We assume [H1] and [H2]. 
%We assume that the single-site potential $u_0 \geq 0$ satisfies $u_0^2 \geq \kappa \chi_0$, for some $\kappa > 0$.
Let $I = [a,b]$ with $b \leq (n+1)^2 \pi^2$, for some $n \in \N \cup \{ 0 \}$. Then, we have
\beq\label{eq:localDOS1}
{\rm Tr} P_\omega^\Lambda (I) \leq \left(  1 - \frac{b}{ (n+1)^2 \pi^2} \right)^{-1}  \sum_{k \in \tilde{\Lambda}} ~
\sum_{j=0}^n \langle \varphi_{j,k} , P_\omega^\Lambda (I) \varphi_{j,k}   \rangle  \left(     1 - \frac{j^2}{ (n+1)^2} \right).
\eeq
In particular, if $n=0$ so $b < \pi^2$, we have
\beq\label{eq:localDOS2}
{\rm Tr} P_\omega^\Lambda (I) \leq \left(  1 - \frac{b}{ \pi^2} \right)^{-1}  \sum_{k \in \tilde{\Lambda}} ~
\langle \varphi_{0,k} , P_\omega^\Lambda (I) \varphi_{0,k}   \rangle  .
\eeq
\end{theorem}

\begin{proof}
Let $\{ \psi_{j} \}$ be an orthonormal basis of eigenfunction of $H_\omega^\Lambda$ with corresponding eigenvalues $E_j$. Although the eigenvalues are random variables, the randomness does not play a role in Theorem \ref{thm:poincare1}. We begin by expanding the  trace with respect to the orthonormal basis of eigenfunctions $\{ \psi_j \}$ and use the decomposition $\chi_{\Lambda} = \sum_{k \in \widetilde{\Lambda}} \chi_k$ of the identity on $\Lambda$  giving
\beq\label{eq:decomp11}
{\rm Tr} P_\omega^\Lambda (I) = \sum_{\{j: E_j \in I\}} \sum_{k \in \widetilde{\Lambda}} {\rm Tr} \Pi_{\psi_j} \chi_k  = \sum_{\{j: E_j \in I\}} \sum_{k \in \widetilde{\Lambda}} ~\int_{C_k} ~|\psi_j(x)|^2 ~dx.
\eeq
Assuming Lemma \ref{lemma:lowerbdcube1}, the proof now easily follows by summation over eigenvalues and over lattice points $k \in \tilde{\Lambda}$. The self-adjoint boundary conditions of $H_\omega^\Lambda$ guarantee that the sum $\sum_{k \in \widetilde{\Lambda}} B_k(\psi_E) = 0$, where the boundary term associated with $C_k$, $B_k(\psi_E)$, is defined in \eqref{eq:L2bound1}.
\end{proof}

We now turn to Lemma \ref{lemma:lowerbdcube1} and its proof that is based on a Poincar\'e-type inequality \eqref{eq:spPoincare1}.

\begin{lemma}\label{lemma:lowerbdcube1}
We assume [H1] and [H2].  Let $\psi_E$ be a normalized eigenfunction of $H_\omega^\Lambda$ with eigenvalue $E \in [0, (n+1)^2 \pi^2]$, for some $n \in \N$. Then, for all $k \in \tilde{\Lambda}$, we have
\bea\label{eq:L2bound1}
 \int_{C_k} |\psi_E (x)|^2 ~ dx  & \leq &  \left( 1 + \frac{ \kappa \omega_k- E}{(n+1)^2 \pi^2} \right)^{-1}  \nonumber \\
 &  \times &  \left[ \sum_{j=0}^n | \langle \psi_E, \varphi_{j,k} \rangle|^2 \left( 1 - \frac{j^2}{(n+1)^2} \right) + \frac{B_k(\psi_E)}{(n+1)^2 \pi^2} \right] ,
\eea
where $\kappa > 0$ is the constant in [H1], and the boundary terms $B_k (\psi_E)$ given by 
\beq\label{eq:bk-defn1}
B_k (\psi_E) := \int_{\partial C_k} \psi_E(x) \hat{\nu} \cdot \nabla \psi_E(x), 
\eeq
%The boundary terms $B_k (\psi_E)$ 
satisfy
\beq\label{eq:sum1}
\sum_{k \in \widetilde{\Lambda}} B_k(\psi_E) = 0.
\eeq
\end{lemma}

\begin{proof}
1. Working with $k=0$ for simplicity, we define the projector $P_n$ by
\beq
P_n\psi_E := \chi_0 \psi_E - \sum_{j=0}^n \langle \varphi_{j,0}, \psi_E \rangle \psi_{j,0}.
\eeq
The vector $P_n \psi_E$ is the projection of $\chi_0 \psi_E$ onto the spectral subspace of $h_{0,0}$ spanned by eigenstates of $h_{0,0}$
with energy at least $(n+1)^2 \pi^2$.
As a consequence, we have the Poincar\'e-type inequality for $P_n \psi_E$:
\beq\label{eq:spPoincare1}
\int_{C_0} | P_n \psi_E|^2 ~dx \leq \frac{1}{(n+1)^2 \pi^2} \int_{C_0} |\nabla P_n \psi_E |^2 ~dx .
\eeq
%This inequality follows from the fact that $P_n \psi_E$ is the projection of $\chi_0 \psi_E$ onto the spectral subspace of $h_{0,0}$ with energy at least $(n+1)^2 \pi^2$. 
This inequality follows from the expansion of $P_n \psi_E$ in the orthonormal basis $\{ \varphi_{j,0} \}$ of eigenfunctions of $h_{0,0}$ and noting that $\int_{C_0} | \nabla \varphi_{j,0}|^2 = E_{j,0}$. Consequently, we obtain
$$
  \int_{C_0} |\nabla P_n \psi_E |^2 ~dx  =   \sum_{j=n+1}^\infty E_{j,0} | \langle \psi_E, \varphi_{j,0} \rangle |^2  \geq (n+1)^2 \pi^2 
\int_{C_0} | P_n \psi_E|^2,
$$
from which \eqref{eq:spPoincare1} follows.

%Although $\chi_0 \psi_E$ is only in the form domain of $h_{0,0}$, the result can be obtained by a limiting procedure starting with %a projection onto a larger spectral subspace of $h_{0,0}$.

\noindent
2. Decomposing $\chi_0 \psi_E$ with respect to the basis $\varphi_{j,0}$, and using the Poincar\'e-type inequality \eqref{eq:spPoincare1}, we have
\bea\label{eq:decomp1}
\int_{C_0} | \psi_E(x)|^2 ~dx & = & \int_{C_0} | P_n \psi_E(x)|^2 ~dx + \sum_{j=0}^n | \langle \psi_E, \varphi_{j,0} \rangle |^2  \nonumber \\
 & \leq & \frac{1}{(n+1)^2 \pi^2} \int_{C_0}   | \nabla  P_n \psi_E(x)|^2 ~dx + \sum_{j=0}^n | \langle \psi_E, \varphi_{j,0} \rangle |^2 ,
\eea
and
\beq\label{eq:decomp2}
 \int_{C_0}   | \nabla  P_n \psi_E(x)|^2 ~dx  =  \int_{C_0}   | \nabla  \psi_E(x)|^2 ~dx -   \sum_{j=0}^n  (j \pi )^2 | \langle \psi_E, \varphi_{j,0} \rangle |^2 .
\eeq

\noindent
3. Finally, from the assumptions on $H_0^\Lambda$, integration by parts results in
\bea\label{eq:decomp3}
\int_{C_0} | \nabla \psi_E(x)|^2 ~dx & = & - \int_{C_0} (\Delta \psi_E(x)) \psi_E(x) ~dx + B_0( \psi_E ) \nonumber \\
  & \leq & (E - \kappa \omega_0) \int_{C_0} |\psi_E(x)|^2 ~dx  + B_0( \psi_E ) ,
\eea
where the boundary term $B_0$ is
\beq\label{eq:boundary1}
B_0(\psi_E) := \int_{C_0} ~\nabla \cdot (\psi_E(x) \nabla \psi_E(x) ~dx  = \int_{\partial C_0} ~ \psi_E(x) ~\hat{\nu} \cdot \nabla \psi_E(x) ~dS(x),
\eeq
and $dS$ denotes the surface measure.
%These calculations can be justified using the quadratic form $Q_{\Lambda}$ in \eqref{eq:quadform1} associated with $H_0^\Lambda$.
Using expression \eqref{eq:decomp3} in \eqref{eq:decomp2} we obtain,
\beq\label{eq:decomp4}
 \int_{C_0}   | \nabla  P_n \psi_E(x)|^2 ~dx  \leq   (E - \kappa \omega_0) \int_{C_0} |\psi_E(x)|^2 ~dx  + B_0( \psi_E ) -   \sum_{j=0}^n  (j \pi )^2 | \langle \psi_E, \varphi_{j,0} \rangle |^2 .
\eeq
Substituting the right side of \eqref{eq:decomp4} into the right side of \eqref{eq:decomp1} yields the result \eqref{eq:L2bound1} for $k=0$.

\noindent
4. To verify the second result \eqref{eq:sum1}, we note that the equality on the first line of \eqref{eq:decomp3} holds for any $k \in \widetilde{\Lambda}$ replacing $0$:
\beq\label{eq:decomp5}
\int_{C_k} | \nabla \psi_E(x)|^2 ~dx  =  - \int_{C_k} (\Delta \psi_E(x)) \psi_E(x) ~dx + B_k( \psi_E ) .
\eeq
with $B_k(\psi_E)$ defined as in \eqref{eq:boundary1} with $k \in \widetilde{\Lambda}$ replacing $0$. We note that $\Lambda = {\rm Int} ~\overline{\cup_{k \in \widetilde{\Lambda}} ~C_k}$. Because of the self-adjoint boundary conditions, the quadratic form associated with $H_0^\Lambda$ in \eqref{eq:quadform1} satisfies:
\bea\label{eq:sum11}
Q_\Lambda ( \psi_E)  & =  & \int_{\Lambda} |\nabla \psi_E (x)|^2 ~dx \nonumber \\
  & = &  \sum_{k \in \widetilde{\Lambda}} ~  \int_{C_k} |\nabla \psi_E (x)|^2 ~dx \nonumber \\
 & = &   \sum_{k \in \widetilde{\Lambda}} ~  \int_{C_k}   (- \Delta \psi_E )(x) \psi_E(x) ~dx ,
%\nonumber \\
 %& = &  \sum_{k \in \widetilde{\Lambda}} ~  \int_{C_k} |\nabla \psi_E (x)|^2 ~dx ,
\eea
so comparing \eqref{eq:sum11} with the sum of \eqref{eq:decomp5},  we have
\beq\label{eq:sum2}
\sum_{k \in \widetilde{\Lambda}} B_k(\psi_E) = 0,
\eeq
verifying \eqref{eq:sum1}.
\end{proof}

%%%%%%%%%%%%%%%%%%%%%%%%%%%%%%%%%%%%%%%%%%%%%%%%%%%%%%%%%%%%%%%%%%%%%%%%%%%%%%%%%%%%%%%%%%%%%%%%%%%%%%%%%%%%%%%%%%%%%%%%%%%%%%%

\section{An alternate approach to spectral averaging}\label{sec:spectralAve1}
\setcounter{equation}{0}

In this section, we present an alternate approach to spectral averaging based on analytic perturbation theory, and use it to prove a version of the Birman-Solomyak Theorem connecting the DOS with the spectral shift function.
We consider a one-parameter family of self-adjoint operators $H_\omega := H_0 + \omega u^2$ on a separable Hilbert space $\mathcal{H}$.  We assume that the self-adjoint operator $H_0$ has discrete spectrum, at least locally in a bounded interval $I \subset \R$.
The perturbation $u^2$ is a bounded, nonnegative, self-adjoint operator with $\| u^2 \| \leq 1$, and the variable $\omega \in \R$.
% In this section, we use analytic perturbation theory to establish a spectral averaging result.

\begin{theorem}\label{thm:SpecAveraging1}
Let $I \subset \R$ be a bounded interval and $P_\omega (I)$ be the spectral projector for $I$ and $H_\omega$. Let $\varphi \in \mathcal{H}$ be a normalized vector so $\| \varphi \| = 1$. For any $\tau_1 < \tau_2$,
%with $\tau_2 - \tau_1$ small enough depending on $H_0$, $u$, and $I$,
we have
\beq\label{eq:SpAve1}
\int_{\tau_1}^{\tau_2} \langle \varphi, u P_\omega (I) u  \varphi \rangle ~d \omega \leq |I| \|  \varphi \|_{L^2(\supp u)}^2.
\eeq
\end{theorem}

%\end{document}

\begin{proof}
1. The family $H_\omega$ is a type A analytic family of operators. From the standard results on analytic perturbation theory (see, for example, \cite[chapter VII, section 2]{kato}), there are analytic eigenvalues $E_j(\omega) \in I$, corresponding eigenfunctions $\psi_j(\omega)$,   with $\| \psi_j(\omega)\| = 1$, and rank-one
 eigenprojections $P_j(\omega) = \Pi_{\psi_j(\omega)}$, such that
$$
P_\omega (I) = \sum_{\{ j ~:~ E_j(\omega) \in I\}} P_j(\omega) ,
$$
where the sum over the eigenvalues includes multiplicities. Substituting this into the left side of \eqref{eq:SpAve1}, we obtain
\beq\label{eq:SpAve2}
\int_{\tau_1}^{\tau_2} \langle \varphi, u P_\omega (I) u  \varphi \rangle ~d \omega =   \int_{\tau_1}^{\tau_2} \left[ \sum_{\{ j ~:~ E_j(\omega) \in I\}} ~ \langle \varphi, u P_j (\omega) u  \varphi \rangle \right] ~d \omega .
\eeq

\noindent
2. Concerning the projectors $P_j(\omega)$, an application of the Feynman-Hellman Theorem implies that
\beq\label{eq:monotone1}
P_j (\omega) u^2  P_j(\omega) = E_j'(\omega) P_j(\omega).
\eeq
If we let $A_j := P_j(\omega) u$, we form two self-adjoint, rank-one operators: $A_j A_j^\star = P_j (\omega) u^2  P_j(\omega)$, and $A_j^\star A_j = u P_j (\omega) u$. The operator $A_j^\star A_j$ projects onto $u \psi_j(\omega)$, whereas the operator $A_j A_j^\star$ projects onto
 $\psi_j(\omega)$. We assume that $u \psi_j(\omega) \neq 0$. This follows for local Schr\"odinger operators, for example,  by the unique continuation principle. Since $A_j A_j^\star$ and $A_j^\star A_j$ are self-adjoint and have the same eigenvalues (except possibly $0$), the spectral theorem gives
\beq\label{eq:diagonal1}
A_j^\star A_j = u P_j (\omega) \nu = E_j^\prime(\omega) \widetilde{P}_j(\omega),
\eeq
where $\widetilde{P}_j(\omega)$ projects onto $u \psi_j(\omega)$. %, an orthonormal family of rank-one projectors, projects onto $u \psi_j(\omega)$.

\noindent
3. The positivity of the  left side of \eqref{eq:monotone1} implies that $E_j(\omega)$ is monotone increasing. As a consequence, given $E \in I = [a,b]$, let $\omega_j(E) \in [\tau_1, \tau_2]$  be such that $E_j(\omega_j(E)) = E$, whenever such an $\omega_j(E)$ exists.  We perform a change of variables from $\omega \in [\tau_1, \tau_2] \rightarrow E \in I$. With this change of variables and \eqref{eq:diagonal1},  an arbitary term of the sum on the right side of \eqref{eq:SpAve2} becomes
\bea\label{eq:SpAve3}
\int_{\tau_1}^{\tau_2} \langle \varphi, u P_j(\omega) u \varphi \rangle ~d \omega &=&   \int_{\tau_1}^{\tau_2} \| \widetilde{P}_j(\omega) \varphi \|^2 ~E_j^\prime(\omega) ~ d \omega \nonumber \\
 &=&   \int_{\sup \{ a, E_j(\tau_1) \}}^{\inf \{ b, E_j(\tau_2) \} } ~ \| \widetilde{P}_j(\omega_j(E)) \varphi \|^2 ~d E . \nonumber \\
 & &
\eea

\noindent
4. With respect to the projectors $\widetilde{P}_j(\omega)$, it is easy to check that if $\omega_j(E) \neq \omega_{j^\prime} (E)$, then
\beq\label{eq:orthog1}
\widetilde{P}_j (\omega_j (E)) \widetilde{P}_{j^\prime} (\omega_{j^\prime}(E)) = \delta_{j j^\prime} \widetilde{ P}_j (\omega_j (E))
\eeq
This also holds if $\omega_j(E) = \omega_{j^\prime}(E)$ by construction of the $\widetilde{P}_j(\omega)$ by the reduction process as described in \cite[chapter II, section 2.3]{kato}. % if $\omega_j(E) = \omega_{j^\prime}(E)$.
Let us define $f_j(E)$ by
\beq\label{eq:fj1}
f_j(E) := \| \widetilde{P}_j (\omega_j(E)) \varphi \|^2.
\eeq
From \eqref{eq:diagonal1} and \eqref{eq:SpAve3}, it follows that
\beq\label{eq:spAveUB1}
\int_{\tau_1}^{\tau_2}  \langle \varphi, u P_\omega (I) u \varphi \rangle ~ d \omega \leq \int_a^b ~\left[ \sum_{ \{ j ~|~ \omega_j (E) \in [\tau_1, \tau_2] \} } ~ f_j(E) \right] ~dE .
\eeq
%where
%\beq\label{eq:fj1}
%f_j(E) = \| \widetilde{P}_j (\omega_j(E)) \varphi \|^2.
%\eeq
According to the orthogonality condition \eqref{eq:orthog1}, we have $ \sum_j f_j(E)  \leq   \|  \varphi \|_{L^2(\supp u)}^2$, for all $E \in I = [a,b]$. This bound, together with \eqref{eq:spAveUB1},  proves the result.
\end{proof}

%%%%%
%\end{document}
%%%%%

There is a situation where we can have equality in Theorem \ref{sec:spectralAve1}. This is when the interval $[\tau_1, \tau_2]$ is equal to the real line $\R$. The proof of this requires some basic tools from Birman-Schwinger theory developed, for example, in \cite[Appendix B]{aenss}. These operators require that for all $\omega$ the operators $H_\omega$ are \emph{local} in the sense that if $H_\omega \varphi = 0$ on any open set in $\R^d$, then $\varphi = 0$ on that set. The Schr\"odinger operators considered here are local in this sense.

\begin{corollary}\label{cor:SpAveEquality1}
Assume that $H_\omega = H_0 + \omega u^2$ is a local operator, in the sense above, for all $\omega \in \R$. Assume that $I \subset \R$ is an interval for which $\sigma(H_0) \cap I$ has zero Lebesgue measure (for example, $\sigma (H_0) \cap I$ is discrete). We then have
\beq\label{eq:SpAveEquality1}
\int_{\R}  \langle \varphi, u P_\omega (I) u \varphi \rangle ~ d \omega = |I| \| \varphi \|_{L^2 ( {\rm supp} ~u)}.
\eeq
\end{corollary}

We assume that $E\not\in \sigma(H_\omega)$ and define the Birman-Schwinger kernel by $K_0(E) := u (H_0 - E)^{-1} u$. According to Lemma B.2 of \cite{aenss} the set of $\omega_j(E)$ in \eqref{eq:fj1} are the repeated eigenvalues of $-K_0(E)^{-1}$ considered as a self-adjoint operator on
$L^2 ( {\rm supp} ~u)$. Moreover, the projectors $\widetilde{P}_j(\omega_j(E))$ in \eqref{eq:diagonal1} are a complete set of eigenprojectors for
$K_0(E)^{-1}$. It follows that
$$
\| \varphi \|_{L^2 ( {\rm supp} ~u)}^2 = \sum_j \| \widetilde{P}_j (\omega_j(E)) \varphi \|^2.
$$
Since this holds for almost every $E \in I$, the result follows from \eqref{eq:SpAve3}.

%We now consider the application of Theorem \ref{thm:SpecAveraging1} to the spectral shift function.
Turning to the spectral shift function, from \eqref{eq:spAveUB1}, we recover some known results about the connection between the spectral shift function (SSF) for the pair $(H_0, H_\omega)$ and the local density of states as first proven in \cite{CHK2007, simon1998}.
For any $\varphi \in \mathcal{H}$, we define $\eta_\varphi (E)$ to be
\bea\label{eq:eta1}
\eta_\varphi (E) & := & \lim_{\epsilon \rightarrow 0^+} \frac{1}{\epsilon}  \int_{\tau_1}^{\tau_2}  \langle \varphi, u P_\omega ([E, E+ \epsilon]) u \varphi \rangle ~ d \omega \nonumber \\
 & = &  \lim_{\epsilon \rightarrow 0^+} \frac{1}{\epsilon} \sum_j \int_{\sup \{ E, E_j(\tau_1) \}}^{\inf \{ E+\epsilon, E_j(\tau_2) \} } ~f_j(s ) ~ds  \nonumber \\
 & = & \sum_{j \in \Gamma(E)} ~f_j(E)  ,
\eea
where $f_j(E)$ is defined in \eqref{eq:fj1} and $\Gamma (E)$ is the set of indices defined as follows:
\bea\label{Eq"Gamma1}
j \in \Gamma (E)  & \Leftrightarrow & \exists \omega_j \in [\tau_1, \tau_2] ~{\rm s.\ t.} ~  E_j(\omega_j) = E \nonumber \\
  & \Leftrightarrow & \omega_j(E) \in [{\tau_1}, {\tau_2} ] ,
\eea
so that  
\beq\label{eq:connectSSF1}
{\rm card} ~\Gamma (E) = \xi (E; H_{\tau_2}, H_{\tau_1} ) .
\eeq
That is, the integer ${\rm card} ~\Gamma (E)$  is the number of eigenvalues of $H_\omega$ crossing $E$ as $\omega$ runs from $\tau_1$ to $\tau_2$.

Let $\{ \varphi_n \}_n$ be an orthonormal basis of $\mathcal{H}$ and take $\varphi = \varphi_k$ to be any element. Then, summing  the right side of \eqref{eq:fj1} over this basis and, using the fact that ${\rm Tr} ( \widetilde{P}_j ( \omega_j(E))) = 1$,
the following form of the { \bf{ Birman-Solomyak formula}} now follows from \eqref{eq:eta1} and \eqref{eq:connectSSF1}:\beq\label{eq:ssf2}
 \xi (E; H_{\tau_2}, H_{\tau_1} ) =  \lim_{\epsilon \rightarrow 0^+} \frac{1}{\epsilon} \int_{\tau_1}^{\tau_2} ~ {\rm Tr}( u  P_\omega ([E, E+\epsilon]) u ) ~d \omega .
\eeq

We note that \eqref{eq:ssf2} is a version of the Birman-Solomyak formula established solely by analytic perturbation theory. A similar formula was derived by Simon using the Krein trace formula for resolvents \cite[equation (1)]{simon1998}. A more common version of this formula is
$$
\int_{\tau_1}^{\tau_2}  ~ {\rm Tr} ( u P_\omega (I)u ) ~d \omega =\int_I ~ \xi (E; H_{\tau_2}, H_{\tau_1}) ~dE,
$$
as found, for example, in \cite{CHK2007}.

%We next provide a bound on the $SSF$ the will be used below.
%
%\begin{lemma}\label{lemma:bound_ssf1}
%Under the hypotheses of Corollary \ref{cor:}, the SSF $\xi (E; H_{\tau_2}, H_{\tau_1} )$ satisfies the bound,
%\beq\label{eq:bound_ssf1}
%\xi (E; H_{\tau_2}, H_{\tau_1} ) \leq {\rm Tr} P_{\tau} ([ E - \| u \|^2 ( \tau_2 - \tau_1) , E]).
%\eeq
%\end{lemma}

Formula  \eqref{eq:ssf2} applies to the spectral shift function for local Schr\"odinger operators with discrete spectrum discussed here. We consider a one-parameter family of Schr\"odinger operators $H_\omega := H_0 + \omega u^2$ on $L^2 (\Lambda)$, with $u \geq 0$ satisfying $u \in L^\infty_0(\R^d)$ and  for a parameter $\omega \in \R$. The self-adjoint operator $H_0 $ is given by $H_0 = - \Delta + \sum_{j \in \Z^d \backslash \{ 0 \}} \omega_j u_j$. Let $H_\omega^\Lambda$ denote a self-adjoint restriction of $H_\omega$ to $\Lambda \subset \R^d$, similarly for $H_0$. Then the operators $H_0^\Lambda$ and $H_\omega^\Lambda$ have discrete spectrum for all $\omega \in \R$. The Birman-Solomyak formula applies to the pair  $(H_0^\Lambda, H_\omega^\Lambda)$

We conclude this section with a bound on the $SSF$ that will be used in the proof of  Theorem \ref{thm:LipCont1}. 

\begin{lemma}\label{lemma:bound_ssf1}
Under the hypotheses of Corollary \ref{cor:SpAveEquality1}, the SSF $\xi (E; H_{\tau_2}, H_{\tau_1} )$ for $\tau_2 > \tau_1$ satisfies the bound,
\beq\label{eq:bound_ssf1}
\xi (E; H_{\tau_2}, H_{\tau_1} ) \leq {\rm Tr} P_{\tau_1} ([ E - \| u \|^2 ( \tau_2 - \tau_1) , E]).
\eeq
\end{lemma}

\begin{proof}
Let $E(\omega)$ be an eigenvalue of $H_\omega$ crossing $E$ for some vaue $\omega (E) \in [\tau_1, \tau_2]$. If $H_\omega \psi (\omega) = E(\omega) \psi (\omega)$, with $\| \psi (\omega) \| =1$, then by the Feynman-Hellman Theorem  
we have $E^\prime (\omega) \| u \psi(\omega) \|^2 \leq \| u \|^2$. It follows that 
$$
E - E(\tau_1) \leq \|u \|^2 (\omega (E) - \tau_1) \leq \|u\|^2 ( \tau_2 - \tau_1),
$$
which implies the bound \eqref{eq:bound_ssf1} since 
$$
\xi (E; H_{\tau_2}, H_{\tau_1} ) = {\rm Tr}  P_{\tau_1} ([0, E]) -  P_{\tau_2} ([0, E]) ,
$$ 
as follows from the definition of the SSF.
\end{proof}

%%%%%%%%%%%%%%%%%%%%%%%%%%%%%%%%%%%%%%%%%%%%%%%%%%%%%%%%%%%
%%%%%%%%%%%%%%%%%%%%%%%%%%%%%%%%%%%%%%%%%%%%%%%%%%%%%%%%%%%%%%%%%%%%%%%%%%%%%%%%%%%%%%%%%%%%%%%%%%%%%%%%%%%%%%%%%%%%%%%%%%%%%%

\section{Lipschitz continuity of the local  DOS}\label{sec:LipCont1}
\setcounter{equation}{0}

In this section, we establish local regularity of the finite-volume DOS function $n_\Lambda (E)$ at low energy without a localization assumption for a restricted family of random potentials. We first mention that under the hypothesis of Theorem \ref{thm:poincare1}, we can prove the Wegner estimate with an explicit form of  the constant.
The Wegner estimate for random Schr\"odinger operators with an absolutely continuous single-site probability measure with density $0 \leq \rho \in L^\infty_0(\R)$ has the form
 \beq\label{eq:wegner1}
 \E^{\Lambda} \{ {\rm Tr} P_\omega^\Lambda (I) \} \leq C_W \| \rho \|_\infty |\Lambda| |I|,
 \eeq
 for $I \subset \R$ and a finite constant $C_W > 0$ that is depends upon $I_+ = \max I$. In the next proposition, we give an explicit form of the constant.

\begin{proposition}\label{prop:wegner1}
Assume hypotheses [H1] and [H2]. 
Let $I = [a,b]$ with $b < (n+1)^2 \pi^2$. We then have
\beq\label{eq:wegner2}
  \E^{\Lambda} \{ {\rm Tr} P_\omega^\Lambda (I) \} \leq |I| |\Lambda| \left( \kappa^{-1} \| \rho \|_\infty (n+1)\left[  1 - \frac{b}{(n+1)^2 \pi^2} \right]^{-1}   \right) ,
\eeq
where $\kappa > 0$ is the lower bound in [H1].
\end{proposition}

%The proof of the proposition follows from (  2.1) and the spectral averaging result (3.1) . To estimate the matrix elements  one %uses  the bound  Xk<..(as before).The result follows by summing over j......

The proof of the proposition follows from the bound on the trace of the spectral projector in \eqref{eq:localDOS1} and the spectral averaging result \eqref{eq:SpAve1}. 
In order to apply \eqref{eq:SpAve1}, we use  the bound $\chi_{k} \leq u_k \kappa^{-\frac{1}{2}}$ in the inner products on the right side of  \eqref{eq:localDOS1}. After taking the expectation and spectral averaging, the result follows by summing over $j \in \{ 1, \ldots, n\}$ and $k \in \widetilde{\Lambda}$.

% by first using the bound $\chi_{k} \leq u_k \kappa^{-\frac{1}{2}}$ in the inner product, and then taking the expectation of the %trace as  on the left in \eqref{eq:localDOS1}. The spectral averaging result \eqref{eq:SpAve1} can now be applied to estimate %the matrix elements. The result follows by then summing over $j \in \{ 1, \ldots, n\}$ and $k \in \widetilde{\Lambda}$.

We define the local density of states ($\ell$DOS) function $n_\Lambda (E)$ by
\beq\label{eq:dosf1}
n_\Lambda (E) := \lim_{\epsilon \rightarrow 0^+} \frac{1}{\epsilon |\Lambda|} \E^{\Lambda} \{ {\rm Tr} P_\omega^\Lambda ((E, E + \epsilon]) \}.
\eeq
By the Wegner estimate, \eqref{eq:wegner2}, we have the bound
\beq\label{eq:dosBd1}
n_\Lambda (E) \leq C_W := \left( \kappa^{-1} \| \rho \|_\infty (n+1)\left[  1 - \frac{E}{(n+1)^2 \pi^2} \right]^{-1}   \right) ,
\eeq
for $E < (n+1)^2 \pi^2$.  The local density of states function $n_\Lambda$ is related to the $\ell$DOS measure defined in \eqref{eq:dos-defn1} by
$$
\mu_\Lambda (I) = \int_I ~ n_\Lambda (E) ~dE.
$$
 We next show that $n_\Lambda(E)$ is Lipschitz continuous in $E$ for energies in the interval $[0, E_0(d)]$, where $E_0(d)$ is defined in \eqref{eq:energyBound2}, near the bottom of the deterministic spectrum.

\begin{theorem}\label{thm:LipCont1}
We assume [H3]: The single-site potential $u_0 = \kappa \chi_0$ and the single-site probability measure is the uniform distribution on $[0,1]$ so that $\rho(s) = \chi_{[0,1]}(s)$.
Let $n_\Lambda (E)$ be the $\ell$DOS function for the local Hamiltonian $H_\omega^\Lambda$, where $\Lambda = [0, L]^d$, with $L \in \N$. For any $0 \leq E_1 < E_2 < E_0(d)$, with $E_0(d)$ defined in \eqref{eq:energyBound2}, there exist a  finite constant $K_1 > 0$, depending only on $E_2$ and $d$, so that
\beq\label{eq:LipCont1}
|n_\Lambda (E_2) - \ n_\Lambda (E_1)| \leq \min \left\{ C_W, K_1 |\Lambda| (E_2 - E_1) \right\} ,
%\left[ C + \left( 1 - \frac{E_2}{(n+1)^2 \pi^2} \right)^{-1} \right] (E_2 - E_1).
\eeq
where $C_W > 0$ is given in \eqref{eq:dosBd1}.
\end{theorem}

 \begin{proof}
1. Hypothesis [H3] provides the covering condition $\sum_{k \in \widetilde{\Lambda}} u_k^2 = \kappa^2 \chi_{\Lambda}$, for some $\kappa > 0$. For $E_2 > E_1$, definition \eqref{eq:dosf1} implies that 
\bea\label{eq:dosf2}
\lefteqn{ n_\Lambda (E_2) -  n_\Lambda (E_1) } \nonumber \\
   & =  & \lim_{\epsilon \rightarrow 0^+} \frac{1}{\epsilon |\Lambda|} \E \left\{ {\rm Tr} P_\omega^\Lambda ([E_2, E_2 + \epsilon )  - {\rm Tr} P_\omega^\Lambda ([E_1, E_1 + \epsilon]) \right\}  \nonumber \\
    & = &  \lim_{\epsilon \rightarrow 0^+} \frac{1}{\kappa^2 \epsilon |\Lambda|} \sum_{k \in \widetilde{\Lambda}} ~\E \left\{ {\rm Tr} ( u_k^2 P_\omega^\Lambda ([E_2, E_2 + \epsilon ))  - {\rm Tr} (u_k^2  P_\omega^\Lambda ([E_1, E_1 + \epsilon])) \right\}  \nonumber \\
 & = &  \frac{1}{\kappa^2 |\Lambda|}    \sum_{k \in \widetilde{\Lambda}}    \E_{\omega_k^\perp} \left\{  \lim_{\epsilon \rightarrow 0^+} \frac{1}{\epsilon} \E_{\omega_k} \left\{ {\rm Tr} (u_k^2  P_\omega^\Lambda ([E_2, E_2 + \epsilon ))
- {\rm Tr}( u_k^2  P_\omega^\Lambda ([E_1, E_1 + \epsilon])) \right\} \right\} , \nonumber \\
\eea
where the interchange of the expectation and the limit may be justified by using the uniform bounds on the $\omega_k$-integrals following from  \eqref{eq:ssf2} and Lemma \ref{lemma:bound_ssf1} so that the Dominated Convergence Theorem applies. 

\noindent
2. By the Birman-Solomyak formula presented in \eqref{eq:ssf2}, we write the limit of the expectation with respect to $\omega_k$ in \eqref{eq:dosf2} as
\bea\label{eq:ssf3}
 \lefteqn{ \lim_{\epsilon \rightarrow 0^+} \frac{1}{\epsilon} \E_{\omega_k} \left\{ {\rm Tr}( u_k^2  P_\omega^\Lambda ([E_2, E_2 + \epsilon )) - {\rm Tr} (u_k^2  P_\omega^\Lambda ([E_1, E_1 + \epsilon])) \right\}} \nonumber \\
& = & \xi (E_2; H_{(\omega_k^\perp, \omega_k =0)}^\Lambda,
 H_{(\omega_k^\perp, \omega_k = 1)}^\Lambda ) -   \xi (E_1; H_{(\omega_k^\perp, \omega_k =0)}^\Lambda,
 H_{(\omega_k^\perp, \omega_k= 1)}^\Lambda )  \nonumber \\
 & = & {\rm Tr} P_{(\omega_k^\perp, \omega_k = 0)}^\Lambda ([E_2, E_1] )  - {\rm Tr} P_{(\omega_k^\perp, \omega_k = 1)}^\Lambda ([E_2, E_1 ])  . 
\eea
In order to bound the expectation with respect to $\omega_k^\perp$ of each trace on the last line of \eqref{eq:ssf3}, we use Lemma \ref{lemma:wegner_minus_1} and obtain
\beq\label{eq:ssf4}
\E_{\omega_k^\perp} \left\{
{\rm Tr} P_{(\omega_k^\perp, \omega_k = 0)}^\Lambda ([E_2, E_1] )  - {\rm Tr} P_{(\omega_k^\perp, \omega_k = 1)}^\Lambda ([E_2, E_1 ]) \right\} \leq 
\frac{c(E_2,d)}{\kappa^2} |\Lambda| |E_2 - E_1| , 
\eeq
where $c(b,d)$ is defined in \eqref{eq:Wegner_cnst1}. 
Combining \eqref{eq:ssf3}--\eqref{eq:ssf4}, we obtain the bound
\beq\label{eq:ssf5}
|  n_\Lambda (E_2) -  n_\Lambda (E_1) | \leq \frac{c(E_2,d)}{\kappa^2} |\Lambda| |E_2 - E_1| , 
\eeq
which is the bound \eqref{eq:LipCont1} with $K_1 := \frac{c(E_2,d)}{ \kappa^2}$. 
\end{proof}

\begin{remark}
We note that the estimate \eqref{eq:LipCont1} is not adequate for controlling the infinite-volume limit. One expects that an additional hypothesis, such as localization, would allow the  removal of the volume factor on the right side of \eqref{eq:LipCont1}. Indeed, during the completion of this note, a preprint of Dolai, Krishna, and Mallick \cite{dkm2019} was posted in which they use localization and obtain Lipschitz continuity of $n_\Lambda (E)$, with a volume independent constant, for $E$ in the region of localization and for a smooth probability density $\rho$. More generally, these authors prove regularity of $n_\Lambda(E)$, depending on the regularity of the single-site probability measure $\rho$, for energies in the region of localization, and obtain regularity results for the infinite-volume limit.
\end{remark}

%%%%%%%%%%%%%%%%%%%%%%%%%%%%%%%%%%%%%%%%%%%%%%%%%%%%%

\begin{appendices}

 \section{Appendix: Some technical results}\label{app:technical1}
\setcounter{equation}{0}

We begin with an estimate on the $L^2$-norm of an eigenfunction of $H_\omega^\Lambda$ restricted to a unit cube $C_0$. 

\begin{lemma}\label{lemma:ef_bound1}
Assume [H1] and [H2] and that $H_\omega^\Lambda$ has Dirichlet boundary conditions on $\Lambda = [-L, L]^d$. Let $\psi_E$ be an eigenfunction of $H_\omega^\Lambda$ with eigenvalue $E$: $H_\omega^\Lambda \psi_E = E \psi_E$, with $\| \psi_E \| =1$ and $E > 0$. Then, for  $C_0$ the unit cube, we have 
\beq\label{eq:ef_bound1}
\| \psi \|_{L^2(C_0)} \leq \frac{d + 4 E}{2d}.
\eeq
\end{lemma}

\begin{proof}
This result follows by an integration by parts. We write $x = ( X_k, x_k)$, where $X_k \in [-L,L]^{d-1}$, and $x_k \in [-L, L]$, and we denote a smaller domain by $\mathcal{T}_k := \{  (X_k, x_k) ~|~ x_k \in [0,1], X_k \in  [-L,L]^{d-1} \}.$
Because of the Dirichlet boundary conditions, we have 
\bea\label{eq:ibp1}
\psi_E (X_k, x_k)^2 &=& \int_{-L}^{x_k}  \frac{\partial }{\partial \tau } \left( \psi_E (X_k,\tau)^2 \right) ~d \tau  \nonumber \\
   & = & 2 \int_{-L}^{x_k} \psi_E (X_k,\tau) \frac{\partial \psi_E}{\partial \tau } (X_k,\tau) ~d \tau  \nonumber \\
          & \leq &  {2} \left(  \int_{-L}^{x_k} \psi_E (X_k,\tau)^2 {d \tau } \right)^{\frac{1}{2}} 
                      \left( \int_{-L}^{x_k} \left( \frac{\partial \psi_E}{\partial \tau} (X_k,\tau) \right)^2  ~d \tau \right)^{\frac{1}{2}} \nonumber \\
   & \leq & {2}  \left(  \int_{-L}^L \psi_E (X_k,\tau)^2 {d \tau } \right)^{\frac{1}{2}} 
                      \left( \int_{-L}^L \left( - \frac{\partial^2 \psi_E}{\partial \tau^2}  (X_k,\tau) \right) \psi_E (X_k, \tau)  ~d \tau \right)^{\frac{1}{2}}  \nonumber \\
 & \leq & \frac{1}{2}  \left(     \int_{-L}^L \psi_E (X_k,\tau)^2 {d \tau }
                  +     4  \int_{-L}^L \left( - \frac{\partial^2 \psi_E}{\partial \tau^2}  (X_k,\tau) \right) \psi_E (X_k, \tau)  ~d \tau \right) . 
   \nonumber \\
 & & 
\eea
Integrating each term in the last line of \eqref{eq:ibp1} over  $\mathcal{T}_k$ (recalling that $x_k \in [0,1]$), we obtain
\beq\label{eq:ibp2}
\int_{\mathcal{T}_k} \psi_E (X_k, x_k)^2 ~dX_k dx_k  \leq  \frac{1}{2} \left( 1 +  4 \left\langle  \left( - \frac{\partial^2 \psi_E}{\partial x_k^2} \right), \psi_E \right\rangle_{L^2(\Lambda_L)}  \right) . %  (X_k,x_k) \psi_E (X_k, x_k) dX_k ~d \tau 
\eeq
Finally, since $C_0 \subset \mathcal{T}_k$ for any $k \in \{ 1, \ldots, d \}$, it follows from \eqref{eq:ibp2} and the positivity of the potential $V_\omega^\Lambda$, that 
\beq\label{eq:final_bd1}
\int_{C_0} | \psi_E(x)|^2 ~dx \leq \frac{1}{d} \sum_{k=0}^d \int_{\mathcal{T}_k}  | \psi_E(x)|^2 ~dx \leq 
\frac{1}{2d} \left( d +  4 \left\langle   H_\omega^\Lambda \psi_E, \psi_E  \right\rangle_{L^2(\Lambda_L)}  \right) .
\eeq
The result follows directly from this and the eigenvalue equation. 
\end{proof}

%%%%%%%%%%%%%%%%%%%%%%%%%%%%%%%%%%%%%%

We apply Lemma \ref{lemma:ef_bound1} in order to derive a version of the Wegner estimate for a random Hamiltonian with one random variable fixed. A similar result was obtained in  \cite[Lemma 4.2]{cgk2010} for more general situations  but with a less explicit constant.

%%%%%%%%%%%%%%%%%%%%%%%%%%%%%%%%%%%%%%%%%%%%%%%%%%%%%
\begin{lemma}\label{lemma:wegner_minus_1}
Assume [H1] and [H2] and that $H_\omega^\Lambda$ has Dirichlet boundary conditions on $\Lambda = [-L, L]^d$.
Let $\tau \geq 0$ and $I = [a,b]$. Then, there exists an energy $E_0 = E_0 (d) < \pi^2$, and a constant $c(b,d) > 0$, depending only on $b$ and $d$, so that for all $0 < b <   E_0$, one has
\beq\label{eq:traceEst1}
\E_{\omega_0^\perp} \{ {\rm Tr} P_{\omega_0^\perp, \tau} (I) \} \leq c(b,d) \kappa^{-2} | \Lambda| |I|.
\eeq
\end{lemma}

\begin{proof}
Let $\psi_E$ be a normalized eigenfunction of $H_{(\omega0^\perp, \tau)}$ with eigenvalue $E \in [a,b]$. From Lemma \ref{lemma:ef_bound1}, it follows that 
\beq\label{eq:ef_bound2}
|\langle \varphi_{0,0}, \psi_E \rangle | \leq \frac{d+ 4 E}{2d},
\eeq
so that with $I = [a,b]$,
\beq\label{eq:tr_bound1}
\langle \varphi_{0,0}, P_{(\omega_0^\perp, \tau)}^\Lambda (I) \varphi_{0,0} \rangle \leq  \frac{d+ 4b}{2d} {\rm Tr}  P_{(\omega_0^\perp, \tau)}^\Lambda (I)  .
\eeq
We now bound the trace according to Theorem \ref{thm:poincare1},
\bea\label{eq:poincare2}
 {\rm Tr}  P_{(\omega_0^\perp, \tau)}^\Lambda (I)  & \leq &  \left(  1 - \frac{b}{ \pi^2} \right)^{-1}  \sum_{k \in \tilde{\Lambda}\backslash \{ 0 \} } ~  \langle \varphi_{0,k} , P_{(\omega_0^\perp, \tau)}^\Lambda (I) \varphi_{0,k}   \rangle  \nonumber \\
 & & +  \left(  1 - \frac{b}{ \pi^2} \right)^{-1}  \left( \frac{d+4 b}{2d} \right) {\rm Tr}  P_{(\omega_0^\perp, \tau)}^\Lambda (I)  ,
\eea
where we used  \eqref{eq:ef_bound1} for the $\varphi_{0,0}$ term. 
We take $b < \pi^2$ 
%$b < E_0 = \frac{7}{2} <  \pi^2$, 
sufficiently small so that the coefficient of the last trace term on the right in \eqref{eq:poincare2} is bounded above as
\beq\label{eq:energyBound1}
\left(  1 - \frac{b}{ \pi^2} \right)^{-1}  \left( \frac{d+4 b}{2d} \right) < 1,
\eeq
so this term can be moved to the left side. 
%For fixed $m \geq 1$, we define $E_n (d,m) := \sup \{ b < \pi^2 ~|~ \eqref{eq:energyBound1} ~holds \}$.
Condition \eqref{eq:energyBound1} requires that $b$ satisfy:
\beq\label{eq:energyBound2}
0 < b \leq E_0 (d) := \frac{1}{2} \left( \frac{ \pi^2 d}{ 2 \pi^2 + d}  \right) <  \frac{1}{2}  \pi^2 .
\eeq
and the bound approaches $\frac{1}{2}  \pi^2$ as $d$ becomes large. 
%So if $d \geq 2$, the condition $b \leq E_n(d)$ %implies $b < (n+1)^2 \pi^2$, as required.  
%\textbf{The bound I got in Lemma A.1 now requires $b$ to be much smaller.}  
This  results in the bound
\beq\label{eq:poincare3}
 {\rm Tr}  P_{(\omega_0^\perp, \tau)}^\Lambda (I)  \leq  c(b,d) \sum_{k \in \tilde{\Lambda}\backslash \{ 0 \} } ~  \langle \varphi_{0,k} , P_{(\omega_0^\perp, \tau)}^\Lambda (I) \varphi_{0,k}   \rangle ,
\eeq
 where 
\beq\label{eq:Wegner_cnst1}
c(b,d) := \left(  1 -  \frac{b}{\pi^2} - \frac{d + 4b}{2d} \right)^{-1} ,
\eeq  
for $b < E_0$.
 %\tilde{c}_n (b,d))^{-1}$, with $\tilde{c}_n(b,d)$ given by the constant on the left in \eqref{eq:energyBound1}.
%\frac{b}{  \pi^2}  - \frac{1+b}{8d}\right)^{-1}$.  
The result now follows from \eqref{eq:poincare3} by applying the  spectral averaging result in Theorem \ref{thm:SpecAveraging1}.
\end{proof}

\end{appendices}

%%%%%%%%%%%%%%%%%%%%%%%%%%%%%%%%%%%%%%%%%%%%%%%%%%%%%%%%%%%%%%%%%%%%%%%%%%%%%%%%%%%%%%%%%%

\end{document}